\documentclass[11pt]{article} 
\usepackage[margin=1in]{geometry} 
\pdfoutput=1 

\usepackage[usenames,dvipsnames,svgnames,table]{xcolor} 

\usepackage{amsmath,amsfonts,amssymb,amsthm}
\usepackage{mathtools}
\usepackage{braket}

\usepackage{enumitem} 

\usepackage{mmap}
\usepackage[T1]{fontenc}

\usepackage[utf8]{inputenc}
\usepackage{csquotes}
\usepackage[english]{babel}

\usepackage[
    activate={true,nocompatibility},
    final, 
    tracking=true,
    kerning=true,
    spacing=true,
    factor=1100,
    stretch=10,
    shrink=10
    ]{microtype}
\microtypecontext{spacing=nonfrench}

\usepackage[backend=biber,
              style=alphabetic,
        maxbibnames=99,
      minalphanames=3,
      maxalphanames=4,
            backref=true]{biblatex}

\usepackage{hyperref}
\hypersetup{
    pdftitle={QMA vs QCMA}, 
    pdfauthor={Anonymous submission},
    colorlinks=true, 
    linkcolor=blue, 
    citecolor=blue, 
    urlcolor=green 
}

\AtEveryBibitem{
 \clearlist{address}
 \clearfield{date}
 \clearlist{location}
 \clearfield{month}
 \clearfield{series}
 \clearfield{pages}
 \clearlist{organization}
 \clearfield{number}
 \clearlist{language}

 \ifentrytype{book}{}{
  \clearlist{publisher}
  \clearname{editor}
  \clearfield{issn}
  \clearfield{isbn}
  \clearfield{volume}
 }
}

\DeclareFieldFormat{eprint:eccc}{
\mkbibacro{ECCC}\addcolon\space\ifhyperref
    {\href{http://eccc.hpi-web.de/report/#1/}{\nolinkurl{#1}}}
    {\nolinkurl{#1}}}
\DeclareFieldAlias{eprint:ECCC}{eprint:eccc}
\DeclareFieldFormat{eprint:iacr}{Cryptology ePrint\addcolon\space\ifhyperref
    {\href{https://eprint.iacr.org/#1}{\nolinkurl{#1}}}
    {\nolinkurl{#1}}}
\DeclareFieldAlias{eprint:IACR}{eprint:iacr}
\DeclareFieldFormat{eprint:arxiv}{arXiv\addcolon\space\ifhyperref
    {\href{https://arxiv.org/abs/#1}{\nolinkurl{#1}}}
    {\nolinkurl{#1}}}

\DeclareFieldFormat[article]{title}{#1}
\DeclareFieldFormat[inproceedings]{title}{#1}
\DeclareFieldFormat[incollection]{title}{#1}
\DeclareFieldFormat[online]{title}{#1}

\renewbibmacro{in:}{}

\newcommand{\lName}{1}

\newcommand{\donothing}[1]{#1}

\newcommand{\JACM}{\if\lName1\donothing{Journal of the {ACM}}\else{JACM}\fi}
\newcommand{\SICOMP}{\if\lName1\donothing{{SIAM} Journal on Computing}\else{SICOMP}\fi}
\newcommand{\ToC}{\if\lName1\donothing{Theory of Computing}\else{ToC}\fi}
\newcommand{\ToCGS}{\if\lName1\donothing{Theory of Computing Graduate Surveys}\else{ToC}\fi}
\newcommand{\TOCT}{\if\lName1\donothing{{ACM} Transactions on Computation Theory}\else{TOCT}\fi}
\newcommand{\ToIT}{\if\lName1\donothing{{IEEE} Transactions on Information Theory}\else{TOCT}\fi}
\newcommand{\CCjournal}{\if\lName1\donothing{Computational Complexity}\else{CC}\fi}
\newcommand{\CJTCS}{\if\lName1\donothing{Chicago Journal of Theoretical Computer Science}\else{CJTCS}\fi}

\newcommand{\TCS}{\if\lName1\donothing{Theoretical Computer Science}\else{TCS}\fi}
\newcommand{\IPL}{\if\lName1\donothing{Information Processing Letters}\else{IPL}\fi}
\newcommand{\JCSS}{\if\lName1\donothing{Journal of Computer and System Sciences}\else{JCSS}\fi}

\newcommand{\RSA}{\if\lName1\donothing{Random Structures and Algorithms}\else{RSA}\fi}
\newcommand{\JCTA}{\if\lName1\donothing{Journal of Combinatorial Theory, Series A}\else{JCTA}\fi}
\newcommand{\JCTB}{\if\lName1\donothing{Journal of Combinatorial Theory, Series B}\else{JCTB}\fi}
\newcommand{\PJM}{\if\lName1\donothing{Pacific Journal of Mathematics}\else{PJM}\fi}
\newcommand{\QICjournal}{\if\lName1\donothing{Quantum Information and Computation}\else{QIC}\fi}
\newcommand{\IJQI}{\if\lName1\donothing{International Journal of Quantum Information}\else{IJQI}\fi}
\newcommand{\PRA}{\if\lName1\donothing{Physical Review A}\else{PRA}\fi}
\newcommand{\PRL}{\if\lName1\donothing{Physical Review Letters}\else{PRL}\fi}
\newcommand{\VLDB}{\if\lName1\donothing{International Journal on Very Large Data Bases}\else{VLDB}\fi}

\DefineBibliographyStrings{english}{%
  backrefpage = {p{.}},
  backrefpages = {pp{.}},
}

\newtheorem{theorem}{Theorem}

\newtheorem{lemma}[theorem]{Lemma}

\newtheorem{corollary}[theorem]{Corollary}
\newtheorem{definition}[theorem]{Definition}
\newtheorem{conjecture}[theorem]{Conjecture}

\theoremstyle{definition}

\newcommand{\eq}[1]{\hyperref[eq:#1]{(\ref*{eq:#1})}}
\renewcommand{\sec}[1]{\hyperref[sec:#1]{Section~\ref*{sec:#1}}}
\newcommand{\thm}[1]{\hyperref[thm:#1]{Theorem~\ref*{thm:#1}}}
\newcommand{\lem}[1]{\hyperref[lem:#1]{Lemma~\ref*{lem:#1}}}
\newcommand{\defn}[1]{\hyperref[def:#1]{Definition~\ref*{def:#1}}}
\newcommand{\prop}[1]{\hyperref[prop:#1]{Proposition~\ref*{prop:#1}}}
\newcommand{\cor}[1]{\hyperref[cor:#1]{Corollary~\ref*{cor:#1}}}
\newcommand{\fig}[1]{\hyperref[fig:#1]{Figure~\ref*{fig:#1}}}
\newcommand{\tab}[1]{\hyperref[tab:#1]{Table~\ref*{tab:#1}}}
\newcommand{\alg}[1]{\hyperref[alg:#1]{Algorithm~\ref*{alg:#1}}}
\newcommand{\app}[1]{\hyperref[app:#1]{Appendix~\ref*{app:#1}}}
\newcommand{\conj}[1]{\hyperref[conj:#1]{Conjecture~\ref*{conj:#1}}}
\newcommand{\chap}[1]{\hyperref[chap:#1]{Chapter~\ref*{chap:#1}}}

\newcommand{\comment}[1]{\textup{{\color{red}#1}}}

\DeclareMathOperator{\poly}{poly}
\DeclareMathOperator{\polylog}{polylog}

\DeclareMathOperator{\Dom}{Dom}

\newcommand{\B}{\{0,1\}}
\newcommand{\Ba}{\{0,1,*\}}

\newcommand{\tR}{\widetilde{R}}

\newcommand{\clA}{\mathcal{A}}

\newcommand{\clP}{\mathcal{P}}
\newcommand{\eps}{\varepsilon}

\DeclareMathAlphabet{\mathbbold}{U}{bbold}{m}{n}

\DeclareMathOperator*{\E}{\mathbb{E}} 

\DeclareMathOperator{\bN}{\mathbb{N}}
\DeclareMathOperator{\bF}{\mathbb{F}}

\DeclareMathOperator{\supp}{supp}

\DeclareMathOperator{\RS}{RS}

\DeclareMathOperator{\density}{density}
\DeclareMathOperator{\RU}{RU}

\newcommand{\QMA}{\mathsf{QMA}}
\newcommand{\QCMA}{\mathsf{QCMA}}
\newcommand{\MA}{\mathsf{MA}}
\renewcommand{\O}{\mathcal{O}}

\newcommand{\BQP}{\mathsf{BQP}}
\newcommand{\BQPp}{\mathsf{BQP}_{/\mathrm{poly}}}
\newcommand{\BQPq}{\mathsf{BQP}_{/\mathrm{qpoly}}}
\newcommand{\FBQPp}{\mathsf{FBQP}_{/\mathrm{poly}}}
\newcommand{\FBQPq}{\mathsf{FBQP}_{/\mathrm{qpoly}}}

\allowdisplaybreaks

\begin{document}

\title{Oracle separation of QMA and QCMA with bounded adaptivity}

\author{
Shalev Ben{-}David\\
\small University of Waterloo\\
\small \texttt{shalev.b@uwaterloo.ca}
\and
Srijita Kundu\\
\small University of Waterloo\\
\small \texttt{srijita.kundu@uwaterloo.ca}
}


\date{}
\maketitle

\begin{abstract}
We give an oracle separation between $\QMA$ and $\QCMA$ for quantum algorithms that have bounded adaptivity in their oracle queries;
that is, the number of rounds of oracle calls is small, though each round may involve polynomially many queries in parallel. Our oracle construction is a simplified version of the construction used recently by Li, Liu, Pelecanos, and Yamakawa (2023), who showed an oracle separation between $\QMA$ and $\QCMA$ when the quantum algorithms are only allowed to access the oracle classically.
To prove our results, we introduce a property of relations called \emph{slipperiness}, which may be useful for getting a fully general classical oracle separation between $\QMA$ and $\QCMA$.
\end{abstract}

{\small\tableofcontents}
\clearpage

\section{Introduction}
It is a long-standing open problem in quantum complexity theory whether the two possible quantum analogs of the complexity class $\mathsf{NP}$ are equivalent. $\QMA$ is defined as the class of decision problems that are solvable by a polynomial-time quantum algorithm that has access to a polynomial-sized
\emph{quantum} witness, whereas $\QCMA$ is the class of decision problems that are solvable by a polynomial-time quantum algorithm that only has access to the polynomial-sized \emph{classical} witness.
In other words, the question asks: are quantum proofs more
powerful than classical proofs?

While the inclusion $\QCMA \subseteq \QMA$ is easy to see, the question of whether these two classes are actually equal, which was first posed by Aharonov and Naveh \cite{AN02}, remains unanswered. Indeed, an unconditional separation between these classes is beyond currently known techniques.

An easier, but still unsolved, problem is to show an oracle separation between $\QMA$ and $\QCMA$. This is because oracle separations in the Turing machine model can be shown by means of separations in the much simpler model of \emph{query complexity},
where similar separations between complexity classes
are routinely shown (for example, a recent oracle separation between $\BQP$ and $\mathsf{PH}$ was provided in \cite{RT18}).
The problem of finding an oracle separation between
$\QMA$ and $\QCMA$ has been a longstanding focus of the quantum
computing community; it boils down to asking whether quantum
proofs are more powerful than classical proofs in the query model.

\subsection{Previous work}

The first progress on the question of  an oracle separation of $\QMA$ and $\QCMA$ was made by Aaronson and Kuperberg \cite{AK07}, who showed that there is a quantum oracle, i.e., a blackbox unitary, relative to which $\QMA \neq \QCMA$. Later, Fefferman and Kimmel \cite{FK18} showed that the separation also holds under what they called an ``in-place permutation oracle'', which is still inherently quantum. Ideally, we would like to get these separations in the standard model of classical oracles: classical functions that a quantum algorithm may query in superposition. \cite{BFM23} showed separations between $\QMA$ and $\QCMA$ in other non-standard oracle models.

Very recently, there has been some progress on this question, with two different variations of the standard classical oracle model. Natarajan and Nirkhe \cite{NN22} showed an oracle separation relative to a ``distributional oracle''. This essentially means that the classical oracle is drawn from a distribution, which the prover knows, but the specific instance drawn is not known to the prover. Therefore, the witness only depends on the distribution over the oracles, which makes it easier to show $\QCMA$ lower bounds. Following this, \cite{LLPY23} showed a separation with a classical oracle that is fully known to the prover, but assuming the verifier is only allowed to access this classical oracle classically, i.e., the verifier is not allowed to make superposition queries
(this makes the class similar to $\MA$ in terms
of its query power and witness type). This model is also simpler to analyze because whatever information the verifier gets from the oracle by classically querying it, could also have been provided as the classical $\QCMA$ witness. \cite{LLPY23} also gave an alternate construction of a distributional oracle separation, with a simpler proof than \cite{NN22}. Their constructions are based on the relational problem used by Yamakawa and Zhandry \cite{YZ22}, in their result on quantum advantage without structure.

Closely related to the $\QMA$ vs $\QCMA$ question is the $\BQPq$ vs $\BQPp$ question. $\BQPq$ is the class of decision problems that are solvable by a polynomial-time quantum algorithm with access to polynomial-sized quantum \emph{advice}, which depends non-uniformly on the length of inputs, but nothing else. $\BQPq$ is the class of decision problems solvable by a polynomial-time quantum algorithm with access to polynomial-sized classical advice. Most works which have found oracle separations for $\QMA$ vs $\QCMA$ in various oracle models, such as \cite{AK07, NN22, LLPY23}, have also found oracle separations between $\BQPq$ and $\BQPp$ with related constructions in the same oracle models.

The question of the relative power of classical vs quantum advice was recently resolved unconditionally (without oracles) for relational problems by Aaronson, Buhrman and Kretschmer \cite{ABK23}, who showed an unconditional separation between $\FBQPq$ and $\FBQPp$. $\FBQPq$ and $\FBQPp$ are the classes of relational problems analogous to $\BQPq$ and $\BQPp$ respectively. Their result was based on observing that separations between quantum and classical one-way communication complexity can be used to show separations between classical and quantum advice. The reason their result only works for the relation classes is that a separation in one-way communication complexity which satisfies the necessary conditions can only hold for relational problems. The specific relational problem used in \cite{ABK23} is known as the Hidden Matching problem. But as was observed in \cite{LLPY23}, the Yamakawa-Zhandry problem \cite{YZ22} also achieves the required communication separation, and could have been used instead. In light of this, the constructions in \cite{YZ22} can viewed as a way to convert relational separations in one-way communication complexity, which correspond to relational separations for quantum vs classical advice, to separations for decision $\QMA$ vs $\QCMA$, and $\BQPq$ vs $\BQPp$, relative to classically accessible oracles. The construction is not blackbox --- it does not work if the Hidden Matching Problem is used instead of the Yamakawa-Zhandry problem, though it plausibly might work with a parallel repetition of the former.

\subsection{Our results}
Unlike previous work,
prove an oracle separation between $\QMA$ and $\QCMA$
relative to a \emph{bona fide} regular oracle with regular
(quantum) queries. Our catch is, instead, that we only allow
the algorithms \emph{bounded adaptivity}.

Bounded adaptivity means that the number of rounds of queries made by the algorithms is small, although there can be polynomially many queries in each round.
Although our result is not formally stronger than those of
\cite{NN22} and \cite{LLPY23},
we feel our result is intuitively closer to a full $\QMA$-$\QCMA$ separation, as it allows the full power of classical proofs and some of the power of quantum queries. Our main result is formally stated below.

\begin{theorem}\label{thm:TMSep}
There is an oracle $\O\colon\B^*\to\B$ such that
$\QCMA^{\O,r}\ne \QMA^{\O,r}$, for $r=o(\log n/\log\log n)$.
\end{theorem}
In the above statement, $\QMA^{\O,r}$ is the class of decision problems solvable by QMA algorithms that have oracle access to $\O$, and make at most $r$ batches of parallel queries to $\O$; $\QCMA^{\O,r}$ is defined analogously. The parameter $n$ is the efficiency parameter
(so the number of queries is $\poly(n)$).


\begin{theorem}\label{thm:QuerySep}
There is a function family $F=\{F_N\}_{N\in I}$ which is efficiently
computable in 1-round query $\QMA$, but for which the growth rate of
$\mathrm{QCMA}^r(F_N)$ for $r=o(\log\log N/\log\log\log N)$ as $N\to\infty$ is not in $O(\polylog(N))$.
\end{theorem}
We shall formally define the query versions of QMA and QCMA, and the $r$-round QCMA query complexity $\mathrm{QCMA}^r$ later.

Our construction for the query complexity separation is a somewhat simplified version of the construction in \cite{LLPY23}, which is based on the Yamakawa-Zhandry problem. \cite{YZ22} and \cite{Liu22} showed that there exists a relational problem $R_f$, indexed by functions $f:[n]\times\B^m\to \B$, for $m=\Theta(n)$, such that given oracle access to a quantum advice $\ket{z_f}$, a quantum algorithm on any input $x \in \B^n$, and on average over $f$, can find a $u$ such that $(x,u) \in R_f$\footnote{The Yamakawa-Zhandry relation is a $\mathsf{TFNP}$ relation, which means that the $u$-s are of $\poly(n)$ length, and a $u$ such that $(x,u) \in R_f$ exists for every $x$.}. On the other hand, no quantum algorithm can find such an $u$ for most $x$, when given only a classical advice $z_f$, and classical query access to $f$. Using this relation $R_f$, for a subset $E \subseteq \B^n$, we construct the following oracle:
\[ O[f, E](x,u) = \begin{cases}
            1 & \text{ if } (x,u) \in R_{f} \land x \notin E \\ 
            0 & \text{ otherwise}.
        \end{cases} \]
The 1-instances of the problem $F_N$ that will separate $\QMA$ and $\QCMA$ in the query complexity model will be $O[f,\emptyset]$, and the 0-instances will be $O[f,E]$ for $|E| \geq \frac{2}{3}\cdot 2^n$, for a large subset of all functions $f$.
This is essentially the same construction that is used in \cite{LLPY23}, except they also use an additional oracle $G$ for a random function from $\B^n$ to $\B^n$, which $O$ also depends on.

Note that the query complexity lower bound we obtain for QCMA is of a different nature than the one obtained in \cite{LLPY23}:
we need to lower bound (bounded-round) quantum query algorithms
instead of only classical query algorithms, and we focus on the
worst-case rather than average-case setting. In order to get an oracle separation for Turing machines from a separation in query complexity, one needs to use a diagonalization argument;
because our result is set up a bit differently than in previous work,
we reprove the diagonalization argument for our setting in
\app{diag}.


Finally, we emphasize that the bounded adaptivity limitation of our result is because we allow the full power of classical proofs and also quantum queries. If one were to drop the power of classical proofs
(resulting in the class $\BQP$), or if one were to drop the power
of quantum queries (resulting in, essentially, $\MA$),
it would follow from \cite{LLPY23} that close variants of
$F_N$ cannot be solved
even without the bounded-round restriction. We conjecture
their lower bounds apply to $F_N$ as well.


\subsection{Our techniques}
We briefly describe the techniques used to obtain the query complexity result. We start by observing that the oracle $O[f,\emptyset]$ is essentially just a verification oracle for the Yamakawa-Zhandry relation. Therefore, there is a quantum witness and a quantum algorithm that can distinguish $O[f,\emptyset]$ and $O[f,E]$ by using this witness, with only one query, with probability $1-2^{-\Omega(n)}$ over $f$. The witness for the yes instance $O[f,\emptyset]$ is simply the quantum advice for the Yamakawa-Zhandry problem, which finds a $u$ for any $x$ with probability $1-2^{-\Omega(n)}$ over $f$. The quantum algorithm finds a $u$ for a random $x$ using the witness, and queries the oracle. Since the no instances return 0 on any $(x,u)$ for most $x$, this algorithm can distinguish $O[f,\emptyset]$ and $O[f,E]$ for $1-2^{-\Omega(n)}$ fraction of the $f$-s.

We now consider the uniform distribution over these good $f$-s, which has $\Omega(2^n)$ min-entropy. If there was a classical witness function depending on $f$, of size $k$, that made a quantum algorithm accept $O[f,\emptyset]$ for these $f$-s, then there would exist a fixed witness string $w$ that would make $O[f,\emptyset]$ accept for $2^{-k}$ fraction of $f$-s. The quantum algorithm depends on the witness, but if we fix the witness string $w$, the algorithm is fixed, and we can then ignore the dependence of the algorithm on the witness.

We now attempt to remove rounds of the quantum query algorithm,
starting with the first round, while keeping the behavior
of the algorithm the same on as many oracles as possible.
Every time we remove a round, we restrict our attention
to a smaller set of oracles, all of which are consistent
with a growing partial assignment we assume is given to us.
At the end, the quantum algorithm will have no rounds left,
and hence will make no queries; we want the set of oracles $O[f,\emptyset]$
on which the behavior is preserved to be non-empty, because then
we can conclude that the algorithm cannot distinguish
$O[f,\emptyset]$ and $O[f,E]$ for some large erased set $E$
(since it now makes no queries).

To remove the first round of the query algorithm, we
start by considering the
the uniform distribution over the $2^{-k}$ fraction of good $f$-s such that $O[f,\emptyset]$ is accepted by $w$.
This distribution has $\Omega(2^n) - k$ min-entropy, and therefore, by a result of \cite{GPW17, CDG18}, it can be written as a convex combination of finitely many $(l, 1-\delta)$-dense distributions, for some small $l$ and $\delta$. $(l,1-\delta)$-dense distributions are a concept that was first introduced in the context of communication complexity; an $(l, 1-\delta)$-dense distribution over $N$ coordinates in which $l$ coordinates are fixed, and the rest of the coordinates have high min-entropy in every subset. (Here we are using the same terminology from \cite{GPW17, CDG18} for dense distributions; the terminology we use in our actual proof will be slightly different --- see \sec{QCMA-lems}.) We restrict our attention to such a distribution,
and try to preserve the behavior of the quantum algorithm only
within a subset of its support.

Some coordinates are fixed in the $(l, 1-\delta)$ distribution, which make the probability over this distribution  of the event $(x,u) \in R_f$ non-negligible, for some $(x,u)$ pairs. The quantum algorithm can potentially learn a lot about $f$ by querying the oracle $O[f,\emptyset]$ for these pairs. Therefore, we shall fix the coordinates of $f$ that are fixed by $(x,u)$ being in $R_f$. Here is where we use the fact that the Yamakawa-Zhandry relation is what we shall call \emph{slippery}. This essentially means that given a small number of fixed coordinates for $f$, the number $(x,u)$ pairs that have non-negligible probability is not too high, and the number of extra coordinates fixed by these $(x,u)$ pairs being in $R_f$ is also not too high. The Yamakawa-Zhandry relation being slippery essentially follows from it using a code that has good list recoverability properties. (The Hidden Matching relation, or its parallel repetition, are not slippery by this definition, and so our construction does not work with these.)

Using the slippery property, we can increase the size of the partial
assignment by not too much, and via a hybrid-like argument \cite{BBBV97},
we can ensure that the first round of the
quantum algorithm does not learn much
from queries outside this partial assignment. We then
restrict our attention to oracles consistent with this partial
assignment; on those, we can simulate the first round of the algorithm
without making real queries (we simply use the known partial assignment
and guess ``0'' on the rest of the oracle positions, which
are highly unlikely to be 1). This way, we get a quantum
algorithm with one fewer round, which mimics the original algorithm
on a small (but not too small) set of oracles.

Continuing this way, we eliminate all rounds of the algorithm
while still maintaining a non-empty set of oracles on which
the behavior is preserved. Each such oracle can be ``erased'',
turning a $1$-input into a $0$-input, so we only need the final
$0$-round algorithm to preserve the behavior of the original
algorithm on at least one input. Using this technique,
we can handle up to $o(\log n/\log\log n)$ rounds of $O(\poly n)$
non-adaptive quantum queries each.


\subsection{Discussion and further work}
We expect our techniques for the $\QMA$ vs $\QCMA$ separation may also work for a $\BQPq$ vs $\BQPp$ separation with boundedly adaptive oracle queries, using the same problem that is described in \cite{LLPY23}. Their oracle in the query complexity setting is given by a random function $G$, which the $\BQP$ algorithm has to compute given oracle access to
\[ O[f,G](x,u) = \begin{cases} G(x) & \text{ if } (x,u) \in R'_f \\ \bot & \text{ otherwise,} \end{cases} \]
and a quantum or classical advice. Here $R'_f$ is a modified $1$-out-of-$n$ version of the Yamakawa-Zhandry problem, which has better completeness properties, but is similar to the original problem otherwise. Clearly this problem can be solved in $\BQPq$ by using the quantum advice for the Yamakawa-Zhandry problem. It cannot be solved on input $x$ with any classical advice and with access to an oracle that outputs $\bot$ for every $(x,u)$. In order to show a $\BQPp$ lower bound for this problem, one needs that there exist many $x$-s such that a quantum algorithm with classical advice cannot distinguish $O[f,G]$ from a version of $O[f,G]$ that is erased on those $x$-s. Since $O[f,G]$ essentially serves as a verification oracle for $R'_f$, we expect that when the quantum algorithm has bounded rounds, a proof very similar to our $\QCMA$ lower bound will work.

The final goal is, of course, to be able to show both these results without a bound on the number of rounds of oracle queries the quantum algorithm makes. As mentioned earlier, we fail to do this because the slipperiness parameters of the relation we picked are not good enough, and our methods would work to separate $\QMA$ and $\QCMA$ with an analogous problem definition where the Yamakawa-Zhandry relation is replaced by a different relation $R_f$ that has the appropriate slipperiness property.

We now expand more on the required strong slipperiness property. Let $R_f$ be a family of $\mathsf{TFNP}$ relations on $\B^n \times \B^m$ indexed by $f \in \B^N$, where $m = \poly(n)$ and $N=\Omega(2^n)$. We further assume $R_f$ satisfies the property that if $(x,u) \in R_f$, then there is a polynomial-sized partial assignment $p$ for $f$ which certifies this, i.e., $(x,u) \in R_f$ $\forall f \supseteq p$. Let $\clP \subseteq \{0,1,*\}^N$ denote the set of polynomial-sized partial assignments for $f$. We define the extended version $\tR$ of the family of relations $R_f$ as follows:
\[ \tR = \{(p,x,u) : p \text{ is the minimal partial assignment s.t. } (x,u) \in R_f \, \forall f\supseteq p\}.\]
Since $p$ is polynomial-sized, if we consider the uniform distribution over $\B^N$, $\Pr[p \subseteq f]$ is exponentially small. Now consider a partial assignment $q$ for $f$ with size at most $s(n)$; we fix the bits in $q$ and generate the other bits of $f$ uniformly at random, which can make the probability of some other partial assignments $p$ non-negligible. The slipperiness property is concerned with the total support of all partial assignments $p$ such that $\Pr[p \subseteq f| q \subseteq f]$ is non-negligible, and $(p, x, u) \in \tR$. We say $\tR$ is $(\eta, s(n), t(n))$-slippery if for all $s(n)$-sized $q$, the total support of all $p$-s such that $\Pr[p \subseteq f| q \subseteq f] \geq \eta$ and $(p,x,u) \in \tR$ is at most $t(n)$. See \defn{slippery} for a more formal definition.

Our techniques show that the following conjecture implies an oracle separation between $\QMA$ and $\QCMA$.
\begin{conjecture}\label{conj:slip-sep}
There exists a family of $\mathsf{TFNP}$ relations $R_f$ such that
\begin{enumerate}
\item There exists a polynomial-time algorithm $\clA$, and for each $f$, a $\poly(n)$-sized quantum state $\ket{z_f}$ such that, given access to $x$ and $\ket{z_f}$, $\clA$ can find $u$ such that $(x,u) \in R_f$, with probability at least $1-2^{-\Omega(n)}$ over uniform $x, f$.
\item There exists a function $s(n) = 2^{o(n)}$ such that for all polynomial functions $p(n)$, the extended relation $\tR$ is $\left(1/p(n), s(n), t(n)\right)$-slippery for some $t(n)$ such that $\log(t(n)) = o(\log(s(n)))$.
\end{enumerate}
\end{conjecture}
Assuming \conj{slip-sep} is true, the oracle function separating $\QMA$ and $\QCMA$ would be distinguishing $O[f,\emptyset]$ and $O[f,E]$, for $|E| \geq \frac{2}{3}\cdot 2^n$, which we have defined earlier, using a relation $R_f$ that satisfies the conjecture. (The Yamakawa-Zhandry relation does not seem to satisfy the conjecture; we can only prove it is $(\eta, s(n), t(n))$-slippery, with $t(n)$ bigger than $s(n)$.)

We further note that any family of relations $R_f$ that satisfies \conj{slip-sep} must give an exponential separation between quantum and randomized one-way communication complexity, with the communication setting being that Alice gets input $f$, Bob gets input $x$, and Bob has to output $u$ such that $(x,u)\in R_f$.\footnote{Strictly speaking, condition 1 of the conjecture only implies that there exists a one-way communication protocol, in which Alice sends the state $\ket{z_f}$, which works on average over $x$ and $f$, whereas we usually require worst-case success in communication complexity. However, we can restrict to the set of $x$ and $f$ for which the algorithm $\clA$ works, in order to get the communication problem.} This is because, if there was a polynomial-sized classical message $w_f$ that Alice could send to Bob in the communication setting, then $w_f$ could also serve as a QCMA proof. Therefore, it seems that the slipperiness condition could also be used for lower-bounding one-way randomized communication complexity (although weaker slipperiness parameters than in the conjecture would also suffice for this).


\section{Preliminaries}

\subsection{QMA and QCMA in query complexity}

In this section, we review the formal definitions of QMA, QCMA,
computationally-efficient QMA, and bounded-round QCMA in the context
of query complexity.

\begin{definition}[Bounded-round quantum query algorithm]
For $r,T,n\in\bN$, give the following definition of a quantum query
algorithm $Q$ acting on $n$ bits, using $r$ rounds, with $T$ queries
in each round. The algorithm will be a
tuple of $r+1$ unitary matrices, $Q=(U_0,U_1,\dots,U_r)$.
These unitary matrices will each act on $T$ ``query-input''
registers of dimension $n$, $T$ ``query-output'' registers of
dimension $2$, an ``output'' register of dimension $2$,
and a work register of arbitrary dimension.

For each $x\in\B^n$,
let $U^x$ be the oracle unitary, which acts on the query-input
and query-output registers by mapping
\[\ket{i_1}\ket{b_1}\ket{i_2}\ket{b_2}\dots\ket{i_T}\ket{b_T}
\to \ket{i_1}\ket{b_1\oplus x_{i_1}}\ket{i_2}\ket{b_2
    \oplus x_{i_2}}\dots\ket{i_T}\ket{b_T\oplus x_{i_T}}\]
for all $i_1,\dots,i_T\in [n]$ and all $b_1,\dots,b_T\in\B$.
We extend $U^x$ to other registers via a Kronecker product with
identity, so that $U^x$ ignores the other registers.

The action of the algorithm $Q$ on input $x\in\B^n$, denoted by
the Bernoulli random variable $Q(x)$, will be the result of
measuring the output register of the state
\[U_rU^xU_{r-1}U^x\dots U^xU_1U^xU_0\ket{\psi_{init}},\]
where $\ket{\psi_{init}}$ is a fixed initial state.
\end{definition}

We will use the term ``$T$-query quantum algorithm''
without referring to the number of rounds to indicate $T$
rounds with $1$ query in each.

\begin{definition}[Query algorithm with witness]
Let $Q$ be a $r$-query quantum algorithm on $n$ bits
with $T$ queries in each round.
For any quantum state $\ket{\phi}$ and any $x\in\B^n$,
let $Q(x,\ket{\phi})$ be the random variable corresponding to the
measured output register after the algorithm terminates, assuming the
initial state contained $\ket{\phi}$ in the work register
(with ancilla padding) instead of being $\ket{\psi_{init}}$.
That is, $Q(x,\ket{\phi})$ is a Bernoulli
random variable corresponding to the measurement outcome of the output
register of the final state
\[U_rU^xU_{r-1}U^x\dots U_1U^xU_0\ket{\phi}\ket{pad}.\]
\end{definition}

\begin{definition}[Query QMA and QCMA]
Let $f$ be a possibly partial Boolean function on $n$ bits,
and let $Q$ be a quantum query algorithm on $n$ bits with $T$ total
queries. We say that $Q$ is a QMA algorithm for $f$ with witness
size $k$ if the following holds:
\begin{enumerate}
\item (Soundness.) For every $x\in f^{-1}(0)$ and every $k$-qubit
state $\ket{\phi}$, we have $\Pr[Q(x,\ket{\phi})=1]\le \epsilon$.
\item (Completeness.) For every $x\in f^{-1}(1)$, there exists
a $k$-qubit state $\ket{\phi}$ such that
$\Pr[Q(x,\ket{\phi})=1]\ge 1-\delta$.
\end{enumerate}
Here, $\epsilon$ and $\delta$ govern the soundness and completeness of
$Q$; by default, we take them both to be $1/3$.
We denote the QMA query complexity of $f$ by
$\mathrm{QMA}_{\epsilon,\delta}(f)$, which is the minimum possible value of
$T+k$ over any QMA algorithm for $f$ with the specified soundness
and completeness.

We say that $Q$ is a QCMA algorithm for $f$ if the same conditions
hold, except with the witness state $\ket{\phi}$ quantifying over only
classical $k$-bit strings in both the soundness and completeness
conditions. We define $\mathrm{QCMA}_{\epsilon,\delta}(f)$ analogously to
$\mathrm{QMA}_{\epsilon,\delta}(f)$, and we omit the subscripts when they are
both $1/3$.
\end{definition}

\begin{definition}[Bounded round query QMA and QCMA]
We define $r$-round QMA and QCMA in exactly the same way as the above definition, except the query algorithms are required to have at most $r$ rounds. We use $\mathrm{QMA}^r_{\eps,\delta}(f)$ and $\mathrm{QCMA}^r_{\eps,\delta}(f)$ to denote the $r$-round QMA and QCMA query complexities of $f$ respectively.
\end{definition}

\begin{definition}[Function family]
A \emph{function family} is an indexed set $F=\{f_n\}_{n\in I}$
where $I\subseteq\bN$ is an infinite set and where
each $f_n$ is a partial Boolean function
$f_n\colon\Dom(f_n)\to\B$ with $\Dom(f_n)\subseteq\B^n$.
\end{definition}

\begin{definition}[Efficiently computable QMA]\label{def:efficientQMA}
Let $F=\{f_n\}_{n\in I}$ be a function family.
We say that $F$ is in \emph{efficiently computable query $\QMA$}
if there is a polynomial-time Turing machine which takes in
the binary encoding $\langle n\rangle$ of a number $n\in I$ and
outputs a QMA verifier $Q$ by explicitly writing out the
unitaries of $Q$ as quantum circuits (with a fixed universal
gate set). The verifier $Q$ must be sound and complete for $f_n$. Efficiently computable bounded-round QMA is defined analogously.

In other words, $\mathrm{QMA}(f_n)$ must be $O(\polylog(n))$,
and moreover, the different algorithms for $f_n$ must
be uniformly generated by a single polynomial-time Turing machine.
\end{definition}


With these definitions, we show in \app{diag} that Theorem \thm{QuerySep} implies Theorem \thm{TMSep}.



\subsection{Error-correcting codes}
A Reed-Solomon error-correcting code $\RS_{q, \gamma, k}$ over $\bF_q$, with degree parameter $0 < k < q-1$ and generator $\gamma \in \bF^*_q$, is defined as
\[ \RS_{q, \gamma, k} = \{(f(\gamma), \ldots f(\gamma^q)): f \in \bF_q[x]_{\deg \leq k}\},\]
where $\bF_q[x]_{\deg \leq k}$ is the set of polynomials over $\bF_q$ of degree at most $k$.

Let $q-1=mn$, for some integers $m$ and $n$. The $m$-folded version $\RS^{(m)}_{q, \gamma, k}$ of $\RS_{q, \gamma, k}$ is a mapping of the code to the larger alphabet $\bF_q^m$ as follows:
\[ \RS^{(m)}_{q, \gamma, k} = \{ ((x_1,\ldots, x_m), \ldots, (x_{q-m},\ldots, x_q)): (x_1,\ldots, x_q) \in \RS_{q, \gamma, k}\}. \]
Note that the alphabet of $\RS^{(m)}_{q,\gamma,k}$ is $\bF_q^m$.

\begin{definition}
We say that a code $C \subseteq \Sigma^n$ is combinatorially $(\zeta, \ell, L)$-list recoverable if for any subsets $S_i \subseteq \Sigma$ such that $|S_i| \leq \ell$, we have,
\[ \left|\{(x_1,\ldots, x_n) \in C: |\{i: x_i \in S_i\}| \geq (1-\zeta)n\}\right| \leq L.\]
\end{definition}

\begin{lemma}[\cite{Rud07, YZ22}]
For a prime power $q$ such that $mn=q-1$, any generator $\gamma \in \bF_q^*$, and degree $k < q-1$, $\RS^{(m)}_{q,\gamma,k}$ is $(\zeta, \ell, q^s)$-list recoverable for some $s \leq m$ if there exists an integer $r$ such that the following inequalities hold:
\begin{align}
(1-\zeta)n(m-s+1) & \geq \left(1+\frac{s}{r}\right)(mn\ell k^s)^{1/(s+1)} \label{eq:LR-1} \\
(r+s)\left(\frac{mn\ell}{k}\right)^{1/(s+1)} & < q. \label{eq:LR-2}
\end{align}
\end{lemma}

\begin{corollary}\label{cor:code-par}
Let $m$ be $\Theta(n)$ integer such that $nm+1=q$ is a prime power. Let $k = \frac{5}{6}mn$ and let $c, d$ be constants. Then $\RS_{q,\gamma,k}$ is $(c\log n/n, 2^{(\log n)^d}, 2^{(\log n)^{d+1}})$-list recoverable.
\end{corollary}

This corollary is proved simply by checking that the equations \eqref{eq:LR-1}--\eqref{eq:LR-2} are satisfied with this choice of parameters. The choice of parameters is in fact the same as those as \cite{YZ22}. Therefore, the above code satisfies the other conditions required for the \cite{YZ22} quantum algorithm to succeed in evaluating the relation $R_{C,f}$ defined in the next section. 

\section{The Yamakawa-Zhandry Problem}

For a function $f: [n]\times\B^m \to \B$ and a linear code $C \subseteq \B^{nm}$, define the relation $R_{C,f} \subseteq \B^n\times\B^{nm}$
\[ R_{C,f} = \{ (x,u) = (x_1\ldots x_n, u_1\ldots u_n) : (u_1\ldots u_n \in C) \land (\forall i \, f(i,u_i)=x_i)\}.\]
We shall typically work with $m=\Theta(n)$. We shall usually work with a fixed code $C$, in which case we shall omit the subscript $C$ from $R_{C,f}$.

Let $\clP \subseteq \Ba^{n2^m}$ denote the set of polynomial-sized partial assignments for functions $f:[n]\times\B^m\to\B$. We define the extended version $\tR_C$ of $\{R_{C,f}\}_f$ over $\clP\times\B^n\times\B^{nm}$ as follows:
\[ \tR_C = \{(p,x,u) : p \text{ is the minimal partial assignment s.t. } (x,u) \in R_{C,f} \, \forall f\supseteq p\}.\]
In particular, $(p, x, u)$ is in $\tR_C$ when $p$ is the partial assignment $(f(i,u_i)=x_i)_i$, which is $n$ bits.

\begin{definition}\label{def:slippery}
Let $\tR_n$ be a sequence of relations on $\clP_n\times\B^n\times\B^{\poly(n)}$, where $\clP_n$ consists of fixed polynomial-sized partial assignments for $N=2^{\Omega(n)}$-bit strings, and $\poly(n)$ is some fixed polynomial. We say $\tR_n$ is $(\eta, s(n), t(n))$-slippery w.r.t. distribution $\mu$ on $f$ if for any partial assignment $q$
on $N$ bits with size at most $s(n)$, if we fix the bits of $q$ in
$f$ and generate the other bits of $f$ according to $\mu$ (conditioned on $q$), we
will have
\[\left|\bigcup_{\substack{(p,x,u)\in \tR_n,\\
\Pr_{f\sim \mu}[p\subseteq f|q\subseteq f]\ge \eta}} \supp(p)\right|\le t(n).\]
We omit mentioning the distribution $\mu$ explicitly if it is the uniform distribution.
\end{definition}

\begin{lemma}\label{lem:slippery}
When $C$ is a code with parameters from \cor{code-par}, then
for any constants $c,d$, $\tR_C$ is $(\frac{1}{n^c}, 2^{(\log n)^d}, 2^{(c+2)(\log n)^{d+1}})$-slippery.
\end{lemma}
\begin{proof}
Let $q$ be a partial assignment of size $2^{(\log n)^d}$. For each $i\in[n]$, let $S_i=\{v: (i,v) \text{ is fixed in } q\}$. Clearly for each $i$, $|S_i| \leq 2^{(\log n)^d}$. By \cor{code-par},
\[ C_q = \left|\{u_1\ldots u_n \in C: |\{i: u_i \in S_i\}| \geq n-c\log n\}\right| \leq 2^{(\log n)^{d+1}}.\]
Let us count the number of $(p,x,u)$ tuples that could be in $\tR_{C}$ conditioned on $q$, for which $u$ is in $C_q$. In fact we only need to count the number of $(x,u)$ pairs that could be in $R_{C,f}$, since $p$ is completely fixed by $x$ and $u$. Each $u$ has at most $c\log n$ many locations that are not fixed by $q$, and $x$ can take any value in those $c\log n$ locations, which means there are $2^{c\log n}$ many possible $x$-s for each $u$. Therefore, the number of $(x,u)$ pairs is $2^{(\log n)^{d+1}}\cdot 2^{c\log n}$. Consider the $(p,x,u)$ corresponding to each such $(x,u)$. Since $x$ has $c\log n$ many locations unfixed by $q$, and $p$ only fixes those locations, we have for each such $p$, $\Pr[p\subseteq f| q \subseteq f] \geq \frac{1}{n^c}$. In fact the $(p,x,u)$ tuples we have counted with $u\in C_q$ are the only ones that satisfy $(p, x, u) \in \tR_C$ conditioned on $q$, and $\Pr[p\subseteq f| q \subseteq f] \geq \frac{1}{n^c}$. Since the total support of each $p$ is $n$, we have,
\[ \left|\bigcup_{\substack{(p,x,u)\in \tR_n,\\
\Pr_{f}[p\subseteq f|q\subseteq f]\ge \frac{1}{n^c}}} \supp(p)\right|\le 2^{(\log n)^{d+1}}\cdot 2^{c\log n}\cdot n \le 2^{(c+2)(\log n)^{d+1}}.\]
\end{proof}

\begin{corollary}\label{cor:slip-mu}
If $\mu$ is a distribution such that for all partial assignments $p$ with $|p|=n$, we have $\mu[p] \leq \frac{k}{2^n}$ (where $\mu[p]$ is the probability mass of strings consistent with $p$), then $\tR_C$ from \lem{slippery} is also $(\frac{k}{n^c}, 2^{(\log n)^d}, 2^{(c+2)(\log n)^{d+1}})$-slippery w.r.t. $\mu$.
\end{corollary}
\begin{proof}
Since $\mu[p]\leq \frac{k}{2^n}$ for all $p$, partial assignments that have probability at least $\frac{k}{n^c}$ against $\mu$ conditioned on $q$ have probability at least $\frac{1}{n^c}$ against the uniform distribution conditioned on $q$. Now we can apply \lem{slippery}.
\end{proof}

\begin{theorem}\label{thm:YZ}
There exists a code $C$ such that
\begin{enumerate}
\item $\tR_C$ is $(\frac{1}{n^c}, 2^{(\log n)^d}, 2^{(c+2)(\log n)^{d+1}})$-slippery for any constant $d$.
\item There exists a quantum advice $\ket{z_f}$ with polynomially many qubits, and a polynomial-time quantum algorithm $\clA$ that has access to $\ket{z_f}, x$, and makes no queries to any oracle, such that for any $x \in \B^n$,
\[ \Pr_{f\sim U}[(u \leftarrow \clA(\ket{z_f}, x))\land((x,u) \in R_{C,f})] \geq 1 - 2^{-\Omega(n)},\]
where the probability is over uniformly random functions $f:[n]\times\B^m \to \B$, and the internal randomness of $\clA$. 
\end{enumerate}
\end{theorem}
\begin{proof}
Item 1 is due to \lem{slippery}. As stated before, the problem $\tR_C$, and the choice of parameters for the code $C$ in \lem{slippery}, is the same as \cite{YZ22}. Therefore, item 2 is due to \cite{YZ22,Liu22}.\footnote{The quantum algorithm in \cite{YZ22} makes some non-adaptive quantum queries (not depending on $x$), and does not take an advice state. The modified algorithm, which instead takes an advice state (which is essentially the state of the algorithm in \cite{YZ22} after its non-adaptive queries) and makes no queries, was described in \cite{Liu22}.}
\end{proof}

\section{QMA vs QCMA}
In this section, we prove \thm{QuerySep}. \thm{effQMA} will define the function $F_N$ and show that it is in QMA, and \thm{QCMA-lbound} will show that it is not in QCMA.

\subsection{Construction and QMA protocol}\label{sec:Fn-def}

Fix a code $C$ for which \thm{YZ} holds,
with $c=\log n$. We shall henceforth refer
to $R_{C,f}$ as only $R_f$ for this $C$. For a subset
$E \subseteq \B^n$, define the oracle
$O[f,E]\colon\B^n\times\B^{nm}\to\B$ as
\[ O[f, E](x,u) = \begin{cases}
            1 & \text{ if } (x,u) \in R_{f} \land x \notin E \\ 
            0 & \text{ otherwise}.
        \end{cases} \]

\begin{theorem}\label{thm:effQMA}
There exists an efficient uniform collection of query QMA protocols
(generated uniformly by a polynomial time Turing machine)
which uses $1$ query and polynomial witness size, and which outputs 0 on all oracles $O[f,E]$ with $|E|\ge (2/3)\cdot 2^n$,
and outputs 1 on $O[f,\emptyset]$ for $1-2^{-\Omega(n)}$ fraction of $f$-s.
\end{theorem}


\begin{proof}
The quantum witness for the algorithm will be quantum advice state for $R_f$ from \thm{YZ}. The quantum algorithm works as follows: it samples a uniformly random $x \in \B^n$, and runs the procedure from \thm{YZ} to find a $u$ such that $(x,u)\in R_f$. Note that this requires no queries to the oracle. Then it queries the oracle at $(x,u)$ and returns the query output. If the oracle is $O[f,\emptyset]$ and the actual state $\ket{z_f}$ from \thm{YZ} is provided as witness, then due to \thm{YZ} we have,
\[ \Pr_{f\sim U}[\clA^{O[f,\emptyset]}(\ket{z_f}) = 1] \geq 1 - 2^{-\Omega(n)}.\]
On the other hand, if the oracle is $O[f,E]$ for $|E| \geq \frac{2}{3}\cdot 2^n$, no matter what witness is provided, and what $u$ is obtained from this witness, the oracle outputs 0 on $(x,u)$ for $\frac{2}{3}$ of the $x$-s. Since the algorithm samples a uniformly random $x$ and queries it with some $u$ for every $f$, we have for every $f$,
\[ \Pr[\clA^{O[f,E]}(\ket{z_f})=1] \leq \frac{1}{3}.\qedhere\]
\end{proof}

\paragraph{Defining the function $F_N$.} We now define the following partial query function with input size
$2^n\times 2^{mn}$:
its $1$-inputs are all the oracles $O[f,\emptyset]$ for which the
algorithm from Theorem \thm{effQMA} accepts with probability at least $2/3$,
and its $0$-inputs are $O[f,E]$ for which $O[f,\emptyset]$ is a $1$-input
and $|E|\ge (2/3)\cdot 2^n$. Note that these oracles correspond
to the inputs ``$x$'' of the query problem. This defines a family
$F_N$ of query tasks with $N=2^n\times 2^{mn}$,
and \thm{effQMA} showed that this family is in
efficiently-computable QMA.

\subsection{Useful notation and lemmas}\label{sec:QCMA-lems}

Recall that a non-adaptive quantum algorithm works on $T$
query-input registers and $T$ query-output registers
plus an additional work register $W$, so that its basis
states look like
\[\ket{i_1}\ket{b_1}\ket{i_2}\ket{b_2}\dots\ket{i_T}\ket{b_T}\ket{W}.\]
To clear up notational clutter, we will use
$\vec{i}\in[N]^T$ to represent a tuple of $T$ indices in
$[N]$. Moreover, for a string $x\in\B^N$ and for 
$\vec{i}\in[N]^T$,
we will define $x_{\vec{i}}\coloneqq
(x_{\vec{i}_1},x_{\vec{i}_2},\dots,
x_{\vec{i}_T})$. The basis states can then be written
$\ket{\vec{i}}\ket{\vec{b}}\ket{W}$, and
the action of the query unitary
$U^x$ to the string $x$ is to map
$\ket{\vec{i}}\ket{\vec{b}}\ket{W}
\to \ket{\vec{i}}\ket{\vec{b}\oplus x_{\vec{i}}}\ket{W}$, extended linearly to the
rest of the space. (Here $\oplus$ denotes the bitwise XOR
of the two strings of length $T$.)

Define $\Pi_{\vec{i}}\coloneqq\ket{\vec{i}}\bra{\vec{i}}
\otimes I_{\vec{b},W}$ to be the projection onto basis states
with $\vec{i}$ in the query-input registers. For $i\in[N]$,
define $\Pi_i\coloneqq\sum_{\vec{i}\ni i}\Pi_{\vec{i}}$
to be the projection onto basis states with $i$ occurring in one
of the query-input registers.
The projections $\Pi_{\vec{i}}$ are onto
orthogonal spaces, though the projections $\Pi_i$ are not.
Observe that 
$\sum_{\vec{i}}\Pi_{\vec{i}}=I$, and that
$\sum_i \Pi_i=\sum_i\sum_{\vec{i}\ni i} \Pi_{\vec{i}}
=\sum_{\vec{i}}\sum_{i\in \vec{i}} \Pi_{\vec{i}}=T\cdot I$.
Moreover, since the oracle unitary $U^x$ does not change
the query-input registers, $U^x$ commutes with both $\Pi_{\vec{i}}$
and $\Pi_i$. Another convenient property is that
if $x_{\vec{i}}=y_{\vec{i}}$ for two strings $x,y\in\B^N$,
then $\Pi_{\vec{i}}(U^x-U^y)=0$;
this holds because both $U^x$ and $U^y$ map
$\ket{\vec{i}}\ket{\vec{b}}\ket{W}$ to the same vector
when $x_{\vec{i}}=y_{\vec{i}}$.
Using these properties, we have the following lemma.

\begin{lemma}[Hybrid argument for nonadaptive queries]
\label{lem:hybrid}
For any strings $x,y\in\B^N$ and any quantum state
$\ket{\psi}=\sum_{\vec{i},\vec{b},W}\alpha_{\vec{i},\vec{b},W}
    \ket{\vec{i}}\ket{\vec{b}}\ket{W}$, we have
\[\|U^x\ket{\psi}-U^y\ket{\psi}\|_2^2
\le 4\sum_{i:x_i\ne y_i}\|\Pi_i\ket{\psi}\|_2^2.\]
\end{lemma}

\begin{proof}
We write the following, with justification afterwards.
\begin{align*}
\|U^x\ket{\psi}-U^y\ket{\psi}\|_2^2
&=\left\|\sum_{\vec{i}}\Pi_{\vec{i}}(U^x-U^y)\ket{\psi}\right\|_2^2
\\&=\sum_{\vec{i}}\|\Pi_{\vec{i}}(U^x-U^y)\ket{\psi}\|_2^2
\\&=\sum_{\vec{i}:x_{\vec{i}}\ne y_{\vec{i}}}
    \|\Pi_{\vec{i}}(U^x-U^y)\ket{\psi}\|_2^2
\\&\le\sum_{\vec{i}}\sum_{i\in\vec{i}:x_i\ne y_i}
    \|\Pi_{\vec{i}}(U^x-U^y)\ket{\psi}\|_2^2
\\&=\sum_{i:x_i\ne y_i}\sum_{\vec{i}\ni i}
    \|\Pi_{\vec{i}}(U^x-U^y)\ket{\psi}\|_2^2
\\&=\sum_{i:x_i\ne y_i}
    \left\|\sum_{\vec{i}\ni i}\Pi_{\vec{i}}
                (U^x-U^y)\ket{\psi}\right\|_2^2
\\&=\sum_{i:x_i\ne y_i}\|\Pi_i(U^x-U^y)\ket{\psi}\|_2^2
\\&=\sum_{i:x_i\ne y_i}\|(U^x-U^y)\Pi_i\ket{\psi}\|_2^2
\\&\le 4\sum_{i:x_i\ne y_i}\|\Pi_i\ket{\psi}\|_2^2.
\end{align*}
In the first line, we used $\sum_{\vec{i}}\Pi_{\vec{i}}=I$.
In the second, we used the orthogonality of the images of
the projections $\Pi_{\vec{i}}$. In the third, we used
$\Pi_{\vec{i}}(U^x-U^y)=0$ when $x_{\vec{i}}=y_{\vec{i}}$.

In the fourth line, we replaced the sum over $\vec{i}$
containing at least one $i$ with $x_i\ne y_i$ with a weighted
sum, where the weight of $\vec{i}$ is the number of
$i\in\vec{i}$ such that $x_i\ne y_i$; this weight is
$0$ when $x_{\vec{i}}=y_{\vec{i}}$ and at least $1$ when
$x_{\vec{i}}\ne y_{\vec{i}}$. This weight can be represented
as a sum over $i\in\vec{i}$ with $x_i\ne y_i$,
since we are counting $\vec{i}$ once for each such $i$ in the tuple.

The fifth line flips the order of the sums, and the sixth
uses orthogonality of the images of $\Pi_{\vec{i}}$ to put
the sum back inside the squared norm. The seventh
line is the definition of $\Pi_i$, and the eighth holds since
$\Pi_i$ commutes with $U^x$ and $U^y$. Finally, the last
line follows either from the triangle inequality,
or from the fact that the spectral norm of $(U^x-U^y)$
is at most $2$ (since $U^x$ and $U^y$ are unitary).
\end{proof}

For an oracle $x\in\B^n$ and a block $B\subseteq[N]$,
use $x[B]$ to denote the string $x$ with queries in $B$ erased;
that is, $x[B]_i=x_i$ if $i\notin B$, and $x[B]_i=0$ for $i\in B$.
Next, we use this hybrid argument in combination with
a Markov inequality to show that if a distribution $\mu$
over $\B^n$ has a set of queries $B\in[N]$ that nearly always
return zero for oracles sampled from $\mu$, then
for any non-adaptive quantum algorithm,
there exists a large set of oracles (measured against $\mu$)
such that the algorithm does not detect whether
any subset of $B$ is erased.

\begin{lemma}[Nonadaptive algorithms don't detect oracle erasures]
\label{lem:erasure}
Fix $\ket{\psi}$ representing the state of a quantum algorithm
before a batch of non-adaptive queries. Let $\mu$ be a distribution
over $\B^N$, and let $\epsilon>0$. Let $B=\{i\in[N]:\Pr_{x\sim\mu}[x_i=1]\le \epsilon\}$. Then there exists a set $S\subseteq\B^N$
such that $\mu[S]\ge 1/2$ and for all $x\in S$ and all subsets
$B_1,B_2\subseteq B$, we have
\[\|U^{x[B_1]}\ket{\psi}-U^{x[B_2]}\ket{\psi}\|_2
\le \sqrt{8\epsilon T}.\]
\end{lemma}

\begin{proof}
We write the following, with justification afterwards.
\begin{align*}
\E_{x\sim\mu}\left[\sum_{i:x_i\ne x[B]_i}
    \|\Pi_i\ket{\psi}\|_2^2\right]
&=\E_{x\sim\mu}\left[\sum_{i\in B} x_i
    \|\Pi_i\ket{\psi}\|_2^2\right]
\\&=\sum_{i\in B}\|\Pi_i\ket{\psi}\|_2^2\E_{x\sim\mu}[x_i]
\\&\le \epsilon \sum_{i\in B}\|\Pi_i\ket{\psi}\|_2^2
\\&\le \epsilon \sum_{i\in [N]}\|\Pi_i\ket{\psi}\|_2^2
\\&=\epsilon \sum_{i\in [N]}\sum_{\vec{i}\ni i}
        \|\Pi_{\vec{i}}\ket{\psi}\|_2^2
\\&=\epsilon T\sum_{\vec{i}}\|\Pi_{\vec{i}}\ket{\psi}\|_2^2
\\&=\epsilon T.
\end{align*}
The first line follows by noting that $x_i\ne x[B]_i$ can only
happen if both $i\in B$ and $x_i=1$; we replace the sum
over $i:x_i\ne x[B]_i$ with the sum over $i\in B$, and
multiply the summand by the indicator for $x_i=1$, which is
$x_i$ itself.

The second line is the result of pushing the expectation inside the
sum, and observing that the norm does not depend on $x$
and can be factored out of the expectation. The third line
follows from the definition of $B$: we know that for all
$i\in B$, the probability of $x_i=1$ is at most $\epsilon$.
The fourth replaces the sum over $B$ with that over $[N]$.
The fifth uses the definition of $\Pi_i$, and exchanges the
sum over $\vec{i}$ with the squared norm using orthogonality.
The sixth line follows by noting that each $\vec{i}$ appears
exactly $T$ times in this double sum.
Finally, the last line follows by pushing the sum inside
the squared norm (using orthogonality), and recalling
that $\sum_{\vec{i}}\Pi_{\vec{i}}=I$, together with the fact
that $\ket{\psi}$ is a unit vector.

Given this bound on the expectation, we can apply Markov's inequality
to conclude that at least half the strings $x$ (weighted by $\mu$)
must satisfy
$\sum_{i:x_i\ne x[B]_i}\|\Pi_i\ket{\psi}\|_2^2\le 2\epsilon T$.
Let $S$ be the set of such strings $x$; then $\mu[S]\ge 1/2$.
Observe that for any $x\in S$ and any $B_1,B_2\subseteq B$, the set
$\{i:x[B_1]_i\ne x[B_2]_i\}$ is a subset of
$\{i:x_i\ne x[B]_i\}$. We now apply \lem{hybrid} to get
\[\|U^{x[B_1]}\ket{\psi}-U^{x[B_2]}\ket{\psi}\|_2^2
\le 4\sum_{i:x[B_1]_i\ne x[B_2]_i}\|\Pi_i\ket{\psi}\|_2^2
\le 4\sum_{i:x_i\ne x[B]_i}\|\Pi_i\ket{\psi}\|_2^2\le 8\epsilon T.\]
The desired result follows by taking square roots.
\end{proof}

We will need some properties of distributions on $\B^N$.
For such a distribution $\mu$, let
$\RU(\mu)\coloneqq \max_{x\in\B^N}\log_2(2^N \mu[x])$
be the max relative entropy of $\mu$ relative
to the uniform distribution. We will generally be interested
in distributions $\mu$ such that $\RU(\mu)$ is small
(say, $\polylog N$), which means that no input
$x\in\B^N$ has probability $\mu[x]$ much larger than $2^{-N}$.

For a partial assignment $p$, let $\mu[p]$ be the probability
mass of strings in $\B^N$ which are consistent with $p$.
Let $|p|$ be the size of $p$ (the number of revealed bits in $p$).
We define the density of $\mu$ to be
$\density(\mu)\coloneqq 1-\max_p \frac{\log_2(2^{|p|}\mu[p])}{|p|}$,
with the maximum taken over partial assignments $p$.
The density of the uniform distribution is $1$.

For a partial assignment $p$, we let $\mu|_p$ denote
the distribution $\mu$ conditioned on the sampled input being
consistent with $p$.

\begin{lemma}[Densification]\label{lem:density}
Let $\mu$ be a distribution over $\B^N$, and let $\delta\in(0,1)$.
Then there exists a partial assignment $p$ such that
\begin{enumerate}
    \item $|p|\le \RU(\mu)/\delta$
    \item $\RU(\mu|_p)\le \RU(\mu)/\delta$
    \item $\density(\mu|_p)>1-\delta$, where the density
    is measured on the bits not fixed by $p$.
\end{enumerate}
\end{lemma}

\begin{proof}
Let $p$ be the largest partial assignment for which
$\mu[p]\ge 2^{-(1-\delta)|p|}$. Then
\[2^{-(1-\delta)|p|}
\le \mu[p]
=\sum_{x\supseteq p} \mu[x]
\le 2^{N-|p|}\cdot 2^{-(N-\RU(\mu))}
=2^{\RU(\mu)-|p|},\]
so $\delta |p|\le \RU(\mu)$, from which the first item
follows. Next,
\[\RU(\mu|_p)
=\max_x \log_2(2^N\mu|_p[x])
=\max_{x\supseteq p} \log_2(2^N\mu[x]/\mu[p])
\le \RU(\mu)+\log_2(1/\mu[p])\]
\[\le \RU(\mu)+\log_2(2^{(1-\delta)|p|})
=\RU(\mu)+(1-\delta)|p|
\le \RU(\mu)+(1-\delta)\RU(\mu)/\delta
=\RU(\mu)/\delta,\]
which gives the second item.
Finally, to upper bound the density of $\mu|_p$, let
$q$ be a partial assignment on a set of indices disjoint
from that of $p$. By the maximality of $p$,
we must have $\mu[p\cup q]< 2^{-(1-\delta)(|p|+|q|)}$.
Now,
\[\log_2(2^{|q|}\mu|_p[q])
=\log_2(2^{|q|}\mu[q\cup p]/\mu[p])
< \log_2(2^{|q|}2^{-(1-\delta)(|p|+|q|)}/2^{-(1-\delta)|p|})
=\delta|q|.\]
From this it follows that $\density(\mu|_p)>1-\delta$,
as desired.
\end{proof}

\subsection{QCMA lower bound}

\begin{theorem}\label{thm:QCMA-lbound}
There is no bounded-round, polynomial-cost QCMA protocol
for the family $F_N$  defined in \sec{Fn-def}.
More formally, consider any family of QCMA protocols for the query
problems $F_N$. If the number of rounds for these
QCMA protocols grows slower than $o(\log\log N/\log\log\log N)$,
then either the number of queries or the witness size must grow like
$\log^{\omega(1)} N$.
\end{theorem}

\begin{proof}
Consider a QCMA verifier for the query task $F_N$.
Let $k=k(N)$ denote the witness size for this verifier;
we can assume for contradiction that $k(N)=O(\polylog(N))=O(\poly n)$.
Since the witness is a classical string, there are only
$2^k$ witnesses over which we quantify. Since each
$1$-input $O[f,\emptyset]$ has some witness accepting it,
we conclude that at least one witness $w$ of size $k$
is a valid witness for at least a $2^{-k}$ fraction
of the $1$-inputs, and hence also for at least a
$2^{-k}(1-2^{-\Omega(n)})$ fraction of all oracles $O[f,\emptyset]$
(including those not in the domain of $F_N$). This is because
the fraction of $f$s for which the quantum algorithm
does not succeed with probability at least $2/3$ is
at most $2^{-\Omega(n)}$. We can assume
$2^{-k}(1-2^{-\Omega(n)})\ge 2^{-2k}$.

Let $S$ be the set of $f$ such that $O[f,\emptyset]$ is accepted
by the algorithm given witness $w$. Let $\mu$ be the uniform
distribution over $S$, and observe that $\RU(\mu)\le 2k$.
Let $Q$ be the quantum algorithm which hard-codes the witness
$w$ into the verifier; then $Q$ accepts all oracles
$O[f,\emptyset]$ for $f\in\supp(\mu)$ and rejects
all oracles $O[f,E]$ if $|E|\ge (2/3)2^n$.

We now proceed by iteratively removing rounds of $Q$.
We define a sequence of quantum algorithms
$Q=Q_0,Q_1,\dots,Q_{r-1},Q_r$, where $Q_\ell$ has $r-\ell$
rounds of $T$ queries each; at the beginning, $Q_0=Q$
has $r$ rounds, and at the end, $Q_r$ makes no queries.
We also define a corresponding sequence of distributions
$\mu=\mu_0,\mu_1,\dots,\mu_{r-1},\mu_r$, each of which
will be uniform over a set of functions $f$; this set will
grow smaller with each round.

To define $(Q_{\ell+1},\mu_{\ell+1})$ given
$(Q_\ell,\mu_\ell)$, we proceed in several steps.

\begin{enumerate}
\item First, use \lem{density} with $\delta=1/n$ to find
a partial assignment $q$ with $|q|\le n\RU(\mu_{\ell})$,
$\RU(\mu_{\ell}|_q)\le n\RU(\mu_{\ell})$,
and with $\mu_{\ell}|_q$ being $(1-\delta)$-dense on the
bits not used by $q$.

\item Second, use \lem{erasure}
with $\epsilon = 1/3200r^2 T$
on the distributions of oracles $O[f,\emptyset]$
when $f$ is sampled from $\mu_{\ell}|_q$. The state $\ket{\psi}$
in the lemma will be the state of the algorithm $Q$ just
before the first batch of $T$ queries.
The lemma gives a set 
$S\subseteq\supp(\mu_{\ell}|_q)$ with $\mu_{\ell}|_q[S]\ge 1/2$.
It has the property that for all $f\in S$
and all sets $B_1,B_2$ containing pairs $(x,u)$
with $\Pr_{f\sim\mu_\ell|_q} [O[f,\emptyset](x,u)=1]\le \epsilon$,
we have
$\|U^{O[f,B_1]}\ket{\psi}-U^{O[f,B_2]}\ket{\psi}\|_2\le 1/20r$.
Condition $\mu_{\ell}|_q$
on the set $S$ to get a distribution $\mu'_\ell$.

Note that $O[f,B_1]$ is an abuse of notation, since
normally we erase inputs $x$ to $f$ from the oracle,
yet $B_1$ is a set of pairs $(x,u)$. We will use this abuse
of notation throughout; if we write $O[f,B]$ where $B$
is a set of pairs, we mean to erase those pairs from the
oracle, while if $B$ is a subset of $\Dom(f)$,
we mean to erase the pairs $(x,u)$ for $x\in B$ and all $u$ from the oracle.

\item Third, use the slippery property
from \cor{slip-mu} on $q$ to conclude that
the number of bits used by partial assignments $p$
for which $(p,x,u)\in \tilde{R}_C$ and
$\Pr_{f\sim\mu'_\ell}[p\subseteq f|q\subseteq f]\ge \epsilon/4$ is small.
Recall that $(p,x,u)\in \tilde{R}_C$ means that the condition
$O[f,\emptyset](x,u)=1$ is equivalent to $p\subseteq f$
for all $f$; such certifying $p$ have $|p|=n$.
\cor{slip-mu} can be applied because $\epsilon/4$ is larger than
$1/n^c$ for $c=\log n$, since we are choosing $r=o(\log n/\log\log n)$
and $T\le O(2^{\log^2 n}/\log n)$.
Now, since $\mu_\ell|_q$ is $(1-\delta)$-dense outside of $q$,
the probability of a partial assignment $p$ against $\mu_\ell|_q$
is at most $2^{\delta|p|}$ times the probability against
the uniform distribution conditioned on $q$.
Here $|p|=n$ and $\delta=1/n$, so the probability
against $\mu_\ell|_q$ is at most twice that against the uniform
distribution conditioned on $q$.
Moving from $\mu_\ell|_q$ to $\mu_\ell'$ conditions
on a set $S$ of probability at least $1/2$, so it can increase
the probability of $p$ by at most
a factor of $2$. Hence the probability of $p$ against $\mu'_\ell$ is overall at most 4 times its probability against the uniform distribution. By \cor{slip-mu}, we conclude the total
number of bits used by partial assignments $p$
for which $\Pr_{f\sim\mu_{\ell}'}[O[f,\emptyset](x,u)=1]\ge \epsilon$ is small.
Let $Z$ be the set of all such bits.

Our next modification to $\mu_\ell'$ will be to fix
the bits in $Z$ to the highest-probability partial assignment
(measured against $\mu_\ell'$), and let $\mu_\ell''$ be
$\mu_\ell'$ conditioned on that partial assignment being
consistent with the sampled function $f$.

\item Finally, set $\mu_{\ell+1}=\mu_\ell''$. Also set $Q_{\ell+1}$
to be the quantum algorithm which is the same as $Q_\ell$,
except that the first batch of queries is made to a fake
oracle instead of a real one. The fake oracle is defined as
follows: on queries $(x,u)$ for which $O[f,\emptyset](x,u)$
is fixed for all $f\in\supp(\mu_{\ell+1})$, return this value
$O[f,\emptyset](x,u)$; on queries $(x,u)$ for which this
value is not fixed for $f\in\supp(\mu_{\ell+1})$, return $0$.
Note that the fake oracle does not depend on the true input
oracle $O[f,E]$, so queries to it can be implemented by $Q_{\ell+1}$
with making queries to the real oracle. This replaces the first
round of $Q_\ell$, so $Q_{\ell+1}$ has one less round.
\end{enumerate}

Our approach will work as follows: we start with a function
$f\in\supp(\mu_r)$, and find a large set $E$ of inputs $x$ for which
$(x,u)$ were not fixed by $\supp(\mu_r)$. We will then argue
that the quantum algorithms in the sequence cannot distinguish
the oracles $O[f,\emptyset]$ and $O[f,E]$. This is clear for
the last algorithm $Q_r$, since it makes no queries. We work
backwards to show that each algorithm $Q_{r-1}$, $Q_{r-2}$,
up to $Q_0$ also cannot substantially distinguish between
these two oracles. This gives a contradiction, since we know
$Q$ accepts the oracle $O[f,\emptyset]$, yet $O[f,E]$ is a $0$-input.

In order to find $f$ and $E$, we first show that
$\RU(\mu_r)$ is not too large. Recall that $\RU(\mu_0)\le 2k$.
We will show that $\RU(\cdot)$ did not increase too much
in each of the $r$ iterations that got us from $\mu_0$ to
$\mu_r$. In one iteration, we defined $\mu_{\ell+1}$
from $\mu_\ell$ in $3$ steps. The first step moved from
$\mu_\ell$ to $\mu_\ell|_q$ with
$\RU(\mu_\ell|_q)\le n\RU(\mu_\ell)$. The second step
conditioned the latter distribution on a set $S$ of probability
mass at least $1/2$, which can only increase $\RU(\cdot)$ by $1$,
so $\RU(\mu_\ell')\le n\RU(\mu_\ell)+1$.

The third step
found the set of all bits fixed in partial assignments $p$
which certify some $(x,u)$ as evaluating to $1$, and picked
the highest-probability partial assignment on those bits.
The maximum increase in $\RU(\cdot)$ is the number of bits
that were fixed in this way. This number comes from
\thm{YZ}, and depends on the number of bits fixed in $q$;
when $|q|=2^{(\log n)^d}$, the number we are looking for is
$2^{(c+2)(\log n)^{d+1}}$, so we can express this as
$2^{(c+2)(\log n)(\log |q|)}$. We had $|q|\le n\RU(\mu_\ell)$
and $c=\log n$. It is not hard to see that this additive
increase dominates $n\RU(\mu_\ell)+1$; assuming everything
is large enough (e.g. $\log n$ is sufficiently large,
and $\RU(\mu_\ell)$ is at least $n^2$, which is without
loss of generality by restricting the original $\mu_0$
to a smaller set if necessary), we can get the final
upper bound $\RU(\mu_{\ell+1})\le 2^{(3+\log n)^2\log\RU(\mu_\ell)}$.

In other words, $\log \RU(\mu_\ell)$ increases by a factor
of at most $(3+\log n)^2$ in each iteration, starting
at $\log\max\{2k,n^2\}\le (3+\log n)\log k$ (assuming $k\ge 2$).
We therefore have $\log \RU(\mu_r)\le (3+\log n)^{2r+1}\log k$.

We next essentially apply another iteration (without the second step) to $\mu_r$.
Using \lem{density}, we find a partial assignment
$q'$ such that $\mu_r|_{q'}$ is $(1-\delta)$-dense outside of $q'$,
with $\delta=1/n$. We then apply \thm{YZ}
to conclude there are few pairs $(x,u)$ with
$\Pr_f[O[f,\emptyset](x,u)=1]\ge 1/2$,
and hence few pairs $(x,u)$ with $\Pr_f[O[f,\emptyset](x,u)=1]=1$
when $f$ is sampled from $\mu_r|_{q'}$;
the number of such pairs is at most $2^{(3+\log n)^{2r+3}\log k}$.
Using $k=O(\poly(n))$ and $r= o(\log n/\log\log n)$,
this means that there are at most $2^{o(n)}$
pairs $(x,u)$ that are fixed to $1$ for all the oracles
$O[f,\emptyset]$ for $f\in \supp(\mu_r|_{q'})$.
Therefore, there are $2^{n-o(n)}$ many inputs $x$
such that for all $u$, the pair $(x,u)$ is not fixed to $1$
by $\supp(\mu_r|_{q'})$. Let $E$ be the set of such $x$;
then $|E|\ge (2/3)2^n$. Let $\hat{f}\in\supp(\mu_r|_{q'})$ be arbitrary.

We now know that $Q$ accepts $O[\hat{f},\emptyset]$ and that
$O[\hat{f},E]$ is a $0$-input. We also know that
$Q_r$ cannot distinguish $O[\hat{f},\emptyset]$ and $O[\hat{f},E]$,
since it makes no queries. Now, let
$B=\{(x,u):x\in E,O[\hat{f},\emptyset](x,u)=1\}$.
Moreover, let $B_\ell$ be the set of pairs $(x,u)$
which had $\Pr_{f\sim\mu_{\ell-1}|_q}[O[f,\emptyset](x,u)=1]\le \epsilon$
in iteration $\ell$ (where $q$ is the partial assignment from step 1 of iteration $\ell$). Note that the pairs not in $B_\ell$
are all fixed in all the oracles in the support of $\mu_{\ell}$,
because we choose values for the bits used by their proving
partial assignments $p$. This means that $B\subseteq B_\ell$
for all $\ell$. Also, let $O_\ell$ be the oracle used by
$Q_\ell$ to simulate the first query batch of $Q_{\ell-1}$.
Recall that $O_\ell(x,u)$ returns $0$ unless $(x,u)$
is fixed to $1$ in all $O[f,\emptyset]$ for
$f\in\supp(\mu_{\ell})$. Since the support of $\mu_\ell$
decreases as a subset in each iteration, the bits fixed
in $\mu_\ell$ are also fixed in $\mu_r$, and hence also agree
with $\hat{f}$. This means that $O_\ell$ can be written
as an erased oracle $O[\hat{f},A_\ell]$ for some
set $A_\ell$ of pairs $(x,u)$ that were not fixed in $\mu_\ell$;
in other words, $A_\ell\subseteq B_\ell$.

We now note the oracle $O[\hat{f},E]$ is the same as
$O[\hat{f},B]$. Additionally, since $B,A_\ell\subseteq B_\ell$,
we have by \lem{erasure},
\[\|U^{O[\hat{f},B]}\ket{\psi}-U^{O[\hat{f},A_\ell]}\ket{\psi}\|_2
\le 1/20r\]
where $\ket{\psi}$ is the state right before the first
query of the algorithm $Q_{\ell-1}$.
This can also be written
\[\|U^{O[\hat{f},E]}\ket{\psi}-U^{O_\ell}\ket{\psi}\|_2
\le 1/20r.\]
Now, applying additional unitary matrices does not change
the $2$-norm, and $Q_\ell$ replaces only the first
query of $Q_{\ell-1}$ with $O_\ell$ and applies
the same unitaries as $Q_{\ell-1}$ in all other rounds.
If we use $Q_\ell(O)$ to denote the final state
of $Q_\ell$ on the oracle $O$, we therefore get
\[\|Q_\ell(O[\hat{f},E])-Q_{\ell-1}(O[\hat{f},E])\|_2\le 1/20r.\]
By triangle inequality, we then get
\[\|Q(O[\hat{f},E])-Q_{r}(O[\hat{f},E])\|_2\le 1/20.\]
Since $\emptyset\subseteq B_\ell$ for all $\ell$, the
same argument also works to show that
\[\|Q(O[\hat{f},\emptyset])-Q_{r}(O[\hat{f},\emptyset])\|_2\le 1/20,\]
and of course we also have
$Q_{r}(O[\hat{f},\emptyset])=Q_{r}(O[\hat{f},E])$
since $Q_r$ makes no queries. A final application of the triangle
inequality gives us
\[\|Q(O[\hat{f},E])-Q(O[\hat{f},\emptyset])\|_2\le 1/10.\]
Since measuring the output qubit of $Q(O[\hat{f},\emptyset])$
gives $1$ with probability at least $2/3$, it is not
hard to show that this implies measuring the output qubit
of $Q(O[\hat{f},E])$ gives $1$ with probability above $1/2$.
This gives the desired contradiction, since the latter is a $0$-input.
\end{proof}

\phantomsection\addcontentsline{toc}{section}{References} 
\renewcommand{\UrlFont}{\ttfamily\small}
\let\oldpath\path
\renewcommand{\path}[1]{\small\oldpath{#1}}
\emergencystretch=1em 
\printbibliography

\appendix

\section{Diagonalization argument}\label{app:diag}

\subsection{QMA and QCMA for Turing machines}
In this section, we formally define $\QMA, \QCMA$, the oracle classes and bounded-round oracle classes corresponding to these.
\begin{definition}[Oracle-querying quantum verifier circuit]
An \emph{oracle-querying quantum verifier circuit (OQQV}
is the following type of quantum circuit. It takes in three
types of input sets of qubits: one set of qubits
representing the input string
$x$; a second set of qubits representing a witness state;
and a third set of ancilla qubits. It has gates from a universal
gate set, but it can additionally use a special oracle gate.
The oracle gate can take in any number $k$ of qubits, and
gives $k$ qubits as output.

For any oracle $\O\colon\B^*\to\B$, the behavior of the oracle gates
in a quantum verifier circuit that is instantiated with oracle $\O$
is as follows: each $k$-qubit basis state is mapped
$\ket{y}\to(-1)^{\O(y)}\ket{y}$ by the $k$-qubit oracle gate.

If $C$ is an OQQV,
$x$ is an input, $\ket{\phi}$ is a witness, and $\O$ is an oracle,
then let $C^{\O}(x,\phi)$ denote the Bernoulli random variable
which is the measurement outcome of the first output qubit
of the circuit $C$ when run on input $x$,
witness $\phi$, and zeroes for the ancilla qubits, assuming the
oracle gates of $C$ apply the oracle $\O$.
\end{definition}

\begin{definition}[Soundness and Completeness]
Let $\O\colon\B^*\to\B$ be an oracle, let $n\in\bN$,
let $C$ be an OQQV with $n$
input qubits, and let $f$ be a partial function from $\B^n$ to $\B$.
\begin{enumerate}
\item We say that $C$ is \emph{$\QMA$-sense sound} for $f$ relative to
$\O$ if for every input $x\in f^{-1}(0)$ and every state $\ket{\phi}$
(on a number of qubits equal to the witness size of $C$),
we have $\Pr[C^{\O}(x,\phi)=1]\le 1/3$. The constant $1/3$ is called
the $\QMA$-sense soundness of $C$.
\item We say that $C$ is \emph{$\QCMA$-sense sound} for $f$
relative to $\O$
if the same condition holds, but only for all classical strings
$\ket{\phi}$ instead of all pure states. Note that $\QMA$-soundness
implies $\QCMA$-soundness.
\item We say that $C$ is \emph{$\QMA$-sense complete} for $f$
relative to $\O$ if for every input $f^{-1}(1)$, there exists
a state $\ket{\phi}$ (on the right number of qubits) such that
$\Pr[C^{\O}(x,\phi)=1]\ge 2/3$. The constant $1-2/3$ is called
the $\QMA$-sense completeness of $C$.
\item We say that $C$ is \emph{$\QCMA$-sense complete} for $f$
relative to $\O$ if the same condition holds with a classical witness:
for every $x\in f^{-1}(1)$, there exists a classical string $\ket{\phi}$
that the circuit accepts. Note that $\QCMA$-sense completeness implies
$\QMA$-sense completeness.
\end{enumerate}
\end{definition}

\begin{definition}[Oracle QMA and QCMA]
A $\QMA$ protocol for a language 
$L\subseteq\B^*$ relative to an oracle $\O$ is a polynomial-time
Turing machine $M$ which, on input $0^n$, outputs an
OQQV $C$ which is $\QMA$-sense
sound and complete for the indicator function of $L\cap \B^n$.
A $\QCMA$ protocol for $L$ relative to $\O$ is the same but with
$\QCMA$-sense soundness and completeness.

The class $\QMA^{\O}\subseteq \mathcal{P}(\B^*)$ is the set of
all languages $L$ for which there is a $\QMA$ protocol
relative to $\O$. Similarly, the class $\QCMA^{\O}$ is the set of all
languages for which there is a $\QCMA$ protocol relative to $\O$.
\end{definition}

Observe that $\QCMA^{\O}\subseteq\QMA^{\O}$. This is because
if we had a $\QCMA$ protocol for $L$ relative to $\O$, we could
modify each OQQV it outputs to make the circuit ``measure'' the witness
before proceeding (by making an untouched copy of each qubit of the
witness).
After this modification, the $\QCMA$-sense soundness will imply
$\QMA$-sense soundness, so the resulting OQQV will be $\QMA$-sense
sound and complete. A Turing machine can implement this modification,
and such a TM will be a $\QMA$ protocol for $L$ relative to $\O$.

\begin{definition}[Bounded round QMA and QCMA]
We define a bounded-round OQQV circuit in the natural way
(a quantum circuit that has polynomially many oracle gates in parallel in each ''round'', and a bounded number of such rounds).

For a function $r\colon\bN\to\bN$, we define
$\QMA^{O,r}$ to be the $r$-bounded-round version of
$\QMA$ relative to the oracle $O$; this measure allows
$\QMA$ protocols which, on input $0^n$, generate
an OQQV circuit which uses at most $r(n)$ rounds of queries
to the oracle and is otherwise a valid $\QMA$ oracle.
We define $\QCMA^{O,r}$ similarly.
\end{definition}

\subsection{From query separation to oracle separation}

\begin{theorem}
\thm{QuerySep} implies \thm{TMSep}.
\end{theorem}

\begin{proof}
Let $F=\{f_n\}_{n\in I}$ be a funcion family which is efficiently
computable in $\QMA^1$ but for which the growth rate of $\QCMA^r(f_n)$
as $n\to\infty$ is not in $O(\polylog(n))$, for $r=o(\log\log n/\log\log\log n)$. Let $R(n)$ be some $o(\log n/\log\log n)$ function. We need to construct
an oracle $\O\colon\B^*\to\B$ and a language $L\subseteq\B^*$
such that $L\in\QMA^{\O, R}$ but $L\notin\QCMA^{\O,R}$.

We interpret the oracle $\O$ as taking as input either a pair
of positive integers $(n,i)$, or else a single integer $n$
(the encoding will specify the formatting unambiguously).
On inputs $n$, the oracle $\O(n)$ behaves like an indicator for the
set $I$, returning $1$ if $n\in I$ and $0$ if $n\notin I$.
On input $(n,i)$, if $n\in I$, the oracle will return
$x^n_i$, where $x^n$ is a specific string in $\Dom(f_n)$
that we will choose later. If $n\notin I$, the oracle returns
arbitrarily (its behavior won't matter). The oracle's behavior
on inputs that are incorrectly formatted is also arbitrary.
This completes the specification of the oracle $\O$, except for
the choice of strings $x^n$ for $n\in I$ (one string from the domain
of each $f_n$ function).

The language $L$ will contain the encodings $\langle n\rangle$
of all $n\in I$ for which $f_n(x^n)=1$. We've now specified
both the language $L$ and the oracle $\O$ except for the choice
of strings $x^n$.

We note that regardless of the choice of strings $x^n$,
we will have $L\in \QMA^{\O,R}$. To see this, observe that using
\defn{efficientQMA} we can get a polynomial-time Turing machine
$M$ which takes in $\langle n\rangle$ and returns a
$\polylog(n)$-sized $R$-round OQQV $C_n$ which,
assuming it runs on an oracle
$\O$ that encodes the string $x$, will accept some witness if $f(x)=1$
and reject all witnesses if $f(x)=0$. This Turing machine
maps $\langle n\rangle$ to an OQQV circuit, but we can convert it
to a classical circuit which maps $\langle n\rangle$ to an OQQV
circuit for all $\langle n \rangle$ of a fixed size -- that is,
for all $n$ in an interval $[2^k,2^{k+1})$. We could then
collapse this circuit-outputting-a-circuit into a single OQQV circuit,
which takes in $\langle n\rangle$ and a witness (and ancillas)
and, after making queries to $\O$in $R$ rounds, decides whether to accept
or reject the witness. Moreover, these OQQV circuits
can be generated uniformly by a Turing machine that takes
a size $0^k$ as input and generates in polynomial time
the circuit which handles all $n\in[2^k,2^{k+1})$.
We can also easily modify these OQQV circuits to have them query
$\O(n)$ to ensure that $n\in I$ before proceeding
(rejecting otherwise). The resulting Turing machine is an $R$-round $\QMA$
protocol for $L$ relative to $\O$, so $L\in\QMA^{\O,R}$.

It remains to select the inputs $x^n$, one per function $f_n$,
in a way that ensures $L\notin\QCMA^{\O,R}$. We do so by diagonalization.
Enumerate all pairs $(M,\alpha)$ where $M$ is a candidate $\QCMA^R$
protocol (i.e.\ a Turing machine that outputs $R$-round OQQV circuits) and
$\alpha$ is a growth rate in $O(\poly(n))$ (we can assume the function
$\alpha(n)$ is always $n^c+c$ for some positive integer $c$, to ensure
that $\alpha(n)$ can be efficiently computable and that there are
countably many such growth rates).

We fix choices $x^n$ using an iterative procedure. At each step
of the iteration, there will be some cutoff $N\in\bN$ such that
$x^n$ has been fixed for all $n<N$, but $x^n$ has not yet been
fixed for all $n\ge N$. Each step $t$ of the procedure
will eliminate the possibility that for the $t$-th pair
$(M,\alpha)$ in our enumeration, $M$ is a $\QCMA^R$ protocol for $L$
relative to $\O$ which runs in $\alpha(k)$ steps on inputs
of size $k$ and produces an OQQV $C_k$ which takes in a witness
of size at most $\alpha(k)$ and makes at most $\alpha(k)$ queries
to the oracle.

To handle the pair $(M,\alpha)$, we find the first $f_n$
such that $n\ge N$ and $\QCMA^r(f_n)>2\alpha(|\langle n\rangle|)$.
Note that $|\langle n\rangle|=O(\log n)$, so
$2\alpha(|\langle n\rangle|)$ is a growth rate in $O(\polylog(n))$;
therefore, there must be infinitely many $f_n$ for which
$\QCMA^r(f_n)$ satisfies this condition.

Run $M(0^k)$ for $\alpha(k)$ steps, where
$k=|\langle n\rangle|$;
if it does not terminate, we consider the pair $(M,\alpha)$
handled, and move to the next pair. We can thus assume
it terminates, so let $C_k$ be the OQQV it outputs. If
$C_k$ takes in witnesses of size more than $\alpha(k)$
or if it makes more than $\alpha(k)$ queries or $r(k)$ many rounds of queries to the oracle,
we again consider the pair $(M,\alpha)$ handled. We can
thus assume $C_k$ uses witnesses of size at most $\alpha(k)$
and makes at most $\alpha(k)$ queries to the oracle in $r(k)$ rounds.

The circuit $C_k$, when given input $\langle n\rangle$,
defines a query $\QCMA^r$ protocol of cost at most $2\alpha(k)$:
it takes in a witness of size at most $\alpha(k)$
and makes at most $\alpha(k)$ queries in $r$ rounds to the oracle.
We note that the behavior of this query algorithm
might depend on the values of $\O$ that are outside
of the input to $f_n$; that is, on the values of $\O$
on oracle queries $\O(m,i)$ for $m\ne n$.
To ensure that the query algorithm is well-defined,
we fix all values of $\O$ that $C_k$ might query,
except for the values at $\O(n,i)$ (which will still
encode an input $x\in\Dom(f_n)$). Note that $C_k$
has finite size and can therefore only call $\O$
with finitely many input wires; we can thus 
fix only finitely many positions of $\O$ when ensuring
$C_k$ gives rise to a well-defined query QCMA algorithm
for $f_n$.

Since we've chosen $n$ so that $\QCMA^r(f_n)>2\alpha(k)$,
and since the cost of this $r$-round QCMA protocol is only $2\alpha(k)$,
it follows that this $r$-round QCMA query protocol is either not sound
or not complete: there is an input $x\in\Dom(f_n)$ on which
this query protocol misbehaves. We will pick this input
as our choice for $x^n$. This will ensure that $C_k$
either does not satisfy soundness for $L\cap\B^k$
or does not satisfy completeness, so $M$ is not a valid
$\QCMA^r$ protocol for $L$ relative to $\O$. We can
then set $N$ to be the new minimum number such that
the oracle $\O$ is unfixed for all $n\ge N$, and then
move on to the next pair $(M,\alpha)$.

This procedure iteratively defines the sequence
$x^n$. Each element in the sequence is eventually defined,
and they never change once defined. Therefore,
we get a well-defined infinite sequence $\{x^n\}_{n\in I}$,
and we conclude that the corresponding oracle $\O$
and language $L$ must satisfy $L\notin \QCMA^{\O,R}$.
\end{proof}

\end{document}